\documentclass[pra,twocolumn,superscriptaddress,nofootinbib]{revtex4-1}

\usepackage[percent]{overpic} 
\usepackage[english]{babel}
\usepackage{amsmath,amssymb,amsthm,amscd}
\usepackage{dsfont,hyperref,ctable,commath,mathtools,thmtools}
\usepackage{mathrsfs}

\DeclareMathOperator*{\argmin}{argmin}
\DeclareMathOperator{\Tr}{Tr}
\DeclareMathOperator{\aff}{aff}
\DeclareMathOperator{\spn}{span}

\DeclareMathOperator{\supp}{supp}

\DeclareMathOperator{\vol}{\mathsf{vol}}
\DeclareMathOperator{\ddi}{\mathtt{ddi}}
\DeclareMathOperator{\rank}{\operatorname{rank}}
\DeclareMathOperator{\jac}{\operatorname{J}}
\DeclareMathOperator{\hes}{\operatorname{H}}

\newcommand{\mcal}[1]{\mathcal{#1}}
\newcommand{\mbb}[1]{\mathbb{#1}}
\newcommand{\mbf}[1]{\mathbf{#1}}
\newcommand{\R}{\mbb{R}}

\newtheorem{dfn}{Definition}
\newtheorem{lmm}{Lemma}
\newtheorem{thm}{Theorem}
\newtheorem{cor}{Corollary}

\begin{document}

\title{Data-driven inference  and observational completeness
  of quantum devices}

\date{\today}

\author{Michele \surname{Dall'Arno}}

\email{cqtmda@nus.edu.sg}

\affiliation{Centre   for  Quantum   Technologies,  National
  University  of  Singapore,  3  Science  Drive  2,  117543,
  Singapore}

\author{Asaph \surname{Ho}}

\affiliation{Centre   for  Quantum   Technologies,  National
  University  of  Singapore,  3  Science  Drive  2,  117543,
  Singapore}

\author{Francesco \surname{Buscemi}}

\email{buscemi@i.nagoya-u.ac.jp}

\affiliation{Graduate  School  of  Informatics,  Nagoya
  University, Chikusa-ku, 464-8601 Nagoya, Japan}

\author{Valerio \surname{Scarani}}

\email{physv@nus.edu.sg}

\affiliation{Centre   for  Quantum   Technologies,  National
  University  of  Singapore,  3  Science  Drive  2,  117543,
  Singapore}
\affiliation{Department of Physics,  National
  University  of  Singapore,  2  Science  Drive  3,  117542,
  Singapore}

\begin{abstract}
  \textit{Data-driven inference} was  recently introduced as
  a protocol that, upon the input  of a set of data, outputs
  a mathematical  description for a physical  device able to
  explain the data.  The device so inferred is automatically
  \textit{self-consistent}, that  is, capable  of generating
  all  given data,  and \textit{least  committal}, that  is,
  consistent   with  a   minimal  superset   of  the   given
  dataset.  When  applied to  the  inference  of an  unknown
  device,  data-driven inference  has been  shown to  output
  always the  ``true'' device whenever the  dataset has been
  produced by means  of an \textit{observationally complete}
  setup,  which   plays  here   the  same  role   played  by
  informationally  complete setups  in conventional  quantum
  tomography.
  
  In  this paper  we  develop a  unified  formalism for  the
  data-driven inference of states  and measurements.  In the
  case  of qubits,  in  particular, we  provide an  explicit
  implementation  of  the  inference protocol  as  a  convex
  programming algorithm  for the machine learning  of states
  and   measurements.     We   also   derive    a   complete
  characterization of observational completeness for general
  systems,  from  which  it   follows  that  only  spherical
  $2$-designs achieve  observational completeness  for qubit
  systems.  This  result provides  symmetric informationally
  complete  sets  and mutually  unbiased  bases  with a  new
  theoretical and operational justification.
\end{abstract}

\maketitle

\textit{Introduction}. --- The state of a physical system is
the description of its properties, i.~e., of the outcomes of
every possible  measurement. Famously, for  quantum systems,
the outcome of most measurement is not deterministic, and so
the state  is statistical information.  It is a  truism that
physical properties  depend on  the degree of  freedom under
study: measuring  the polarisation  of an optical  mode, the
spin  of a  silver  atom, or  the energy  level  of a  bound
electron    in   an    atom,   each    requires   its    own
instrumentation.  In  the   theoretical  modelling,  various
degrees  of freedom  may be  described by  the same  Hilbert
space:   all  of   the   above-mentioned   could  be   ``one
qubit''. The  formalism of quantum state  reconstruction, or
tomography,    is     then    identical    for     all    of
them~\cite{BCDFP09}.     This    level     of    abstraction
notwithstanding,   tomography   relies    on   an   accurate
\textit{calibration} of  the devices: in order  to interpret
the data, one  needs to know which setting of  the device is
translated as  (say) $\sigma_x$  in the  theory. Calibration
requires  the  usage of  known,  or  trusted, devices,  thus
introducing   circularity  and   potential  errors   in  the
assessment.  Cartesians  are  doomed   to  remain  in  doubt
forever; most  of us  trust experienced  experimentalists to
perform enough  checks and  calibrations to be  confident of
their assessment.

Nevertheless, quantum devices are  currently leaving labs to
enter  the market.  A potential  buyer may  not be  able, or
simply  not  be allowed,  to  scrutinize  the physics  of  a
commercial black  box. All she  may be  allowed to do  is to
query it  and see how  it responds.  This is why  the recent
years have  witnessed a  growth in interest  about assessing
devices   (source,  measurement,   channel...)  using   only
observed  statistics,  the  structure  of  the  theory,  and
possibly a  few other statistical assumptions  like the fact
that   successive   queries    sample   the   same   process
(independent-and-identically-distributed,  or i.i.d.).  Most
of  this work  has focused  on devices  that violate  Bell's
inequalities, and has been called \textit{device-independent
  certification}. This  paper is in a  different line, which
has         been          called         \textit{data-driven
  inference}~\cite{DBBT18,DBV18}. The goal is to produce the
\textit{least  committal}  mathematical description,  within
the  theory,  of a  device  that  could have  generated  the
observed statistics.

We first present a  unified formalisation of the data-driven
inference of states and effects (measurement elements). This
inference is explicitly implemented as a convex optimization
algorithm~\cite{BV04}  for  theories with  (hyper)-spherical
state space, respectively  (hyper)-conical effect space. For
these    same   theories,    we    prove   theorems    about
\textit{observational completeness},  the notion  that plays
in data-driven inference a role  analogous to that played by
informational       completeness       in       conventional
tomography~\cite{DBBT18}. Specifically,  we prove  that only
spherical  $2$-designs  achieve observational  completeness.
For the quantum case of the qubit, it follows that symmetric
informationally   complete  sets~\cite{Zau99,   RBSC03}  and
mutually  unbiased  bases  are  thus  provided  with  a  new
operational interpretation.   We conjecture this to  be true
for quantum systems of arbitrary dimension.

\begin{figure}[ht]
  \begin{overpic}[width=.8\columnwidth]{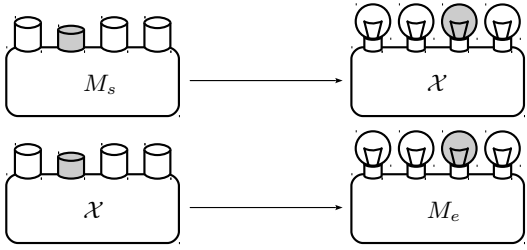}
    \put (15, 5) {$M_s$}
    \put (81, 5) {$\mcal{X}$}
  \end{overpic}
  \\\vspace{2mm}
  \begin{overpic}[width=.8\columnwidth]{fig01}
    \put (15, 5) {$\mcal{X}$}
    \put (81, 5) {$M_e$}
  \end{overpic}
  \caption{Two     ways    of     processing    the     same
    data.     \textbf{Top:}      inference     of     states
    \eqref{eq:states0}: the state  preparator is interpreted
    as  a linear  map $M$  satisfying Eq.~\eqref{eq:states},
    while the effects are represented by a set $\mcal{X}$ of
    vectors.  \textbf{Bottom:}  inference of  a  measurement
    \eqref{eq:meas0}:  the measurement  is interpreted  as a
    linear  map $M$  satisfying Eq.~\eqref{eq:measurements},
    while the states are represented  by a set $\mcal{X}$ of
    vectors.  In  either case,  $M\mcal{X}$  is  the set  of
    probability    vectors    collected   after    (ideally,
    infinitely) many runs.}
  \label{fig:setup}
\end{figure}

\textit{Formalization}.       ---     We      consider     a
prepare-and-measure    scheme    (Figure    \ref{fig:setup})
described in a bilinear  physical theory: the probability of
the  outcome $j\in  [1,..., J]$  when measuring  state $i\in
[1,...,   I]$   is    modelled   by   $p_{ij}=\mbf{e}_j\cdot
\mbf{s}_i=\mbf{e}_j^T    \mbf{s}_i$,   where    the   states
$\mbf{s}_i$ and the effects $\mbf{e}_j$ are (column) vectors
in  a  space  $\mbb{R}^\ell$.   Of  course,  quantum  theory
belongs to  this set  of theories because  of the  Born rule
$p_{ij}=\Tr[\rho_i  E_j]$,  where  $\ell=d^2$ with  $d$  the
Hilbert space dimension.

For the sake of concreteness,  let us provide a paradigmatic
example (detailed in Appendix \ref{app:example}). The source
can produce $I = 3$  states and the measurement is described
by $J = 4$ effects. The \textit{data} are
\begin{align}
  \label{dataexample}
  \mathbf{p} = \left[ \Tr\left[\rho_i E_j\right]\right] =
  \begin{bmatrix}
    \frac12 &  0 & \frac14  & \frac14\\ \frac18 &  \frac38 &
    \frac{2+\sqrt{3}}8  &   \frac{2-\sqrt{3}}8\\  \frac18  &
    \frac38 & \frac{2-\sqrt{3}}8 & \frac{2+\sqrt{3}}8
  \end{bmatrix}.
\end{align}
Since the  rows are  different, we  know trivially  that the
states are  different and that  the effects are  not trivial
(while  a single  row  of data,  i.e.~the  data obtained  by
measuring   a  single   state,   could   always  come   from
$E_j=p_{1j}\openone$).  But with the techniques described in
this paper, one can gather  much more. Indeed, by looking at
the rows,  one can  make the following  \textit{inference on
  the effects}: if  the system is a real  qubit, the effects
are $E_{1,2} = \frac14 (\openone \pm \sigma_z)$ and $E_{3,4}
= \frac14 (\openone  \pm \sigma_x)$ up to  the definition of
these axes in the plane.  By looking at the columns, one can
make the  following \textit{inference on the  states}: again
for a real qubit, the three  states are pure and their Bloch
vectors point at the vertices of an equilateral triangle.

The two inferences have a  very similar formalisation. So we
propose a formal  language applicable to both;  when the two
have to be  differentiated, we shall use  the subscripts $s$
for states and $e$ for effects. To make an \textit{inference
  on the family of states}, we shall study the family of $J$
vectors
\begin{align}
  \label{eq:states0}
  \mbf{x}_{s,j} = \left(p_{1j},  p_{2j}, ..., p_{nj} \right)
  \textrm{ with } n = I,
\end{align}
indexed by  the effect,  whose components are  determined by
the states.  Conversely, to make an \textit{inference on the
  family of  effects (i.~e., on the  measurement)}, we shall
study   the  family   of  $I$   vectors
\begin{align}
  \label{eq:meas0}
  \mbf{x}_{e,i} =  \left(p_{i1}, p_{i2},  ..., p_{in}\right)
  \textrm{ with } n=J,
\end{align}
indexed by the state, whose components are determined by the
measurement.    Compactly:   \textit{a  family   of   states
  (effects) is seen  as a linear map  $M_{s(e)} \in \R^{\ell
    \to n}$ from the space  of effects (states) to the space
  of probabilities}.  Such a linear  map is the object to be
inferred from the dataset.

The   two    maps   defined   by    \eqref{eq:states0}   and
\eqref{eq:meas0}  differ  because $\sum_jp_{ij}=1$  for  all
$i$,   while   $\sum_ip_{ij}$   does   not   obey   such   a
constraint. This  difference has a  geometric interpretation
(Fig.~\ref{fig:spaces}). In $\mathbb{R}^\ell$, let us define
the \textit{unit effect} $\mbf{u}_\ell$, which is the effect
such  that  $\mbf{u}_\ell\cdot  \mbf{s}=1$  for  all  states
$\mbf{s}$. On the one hand, a  family of $n$ states maps the
unit effect onto the vector $\mbf{u}_n\in\mathbb{R}^n$ whose
entries are all ones. Thus,  the map $M_s$ for the inference
of states satisfies
\begin{align}
  \label{eq:states}
  M_s \mbf{u}_\ell = \mbf{u}_n\,.
\end{align} 
On the other hand, a family of $n$ effects maps a state into
a normalised  probability vector \eqref{eq:meas0}:  in other
words, it  maps the hyperplane orthogonal  to $\mbf{u}_\ell$
defined by $\mbf{u}_\ell\cdot \mbf{s}=1$ into the hyperplane
orthogonal   to  $\mbf{u}_n$   defined  by   $\mbf{u}_n\cdot
\mbf{p}=1$. Thus, the map $M_e$ for the inference of effects
satisfies
\begin{align}
  \label{eq:measurements}
  M_e^T \mbf{u}_n = \mbf{u}_\ell\,.
\end{align}

In  fact,  the  actual  choice of  coordinates  for  vectors
$\mbf{u}_\ell$  and  $\mbf{u}_n$  in  Eqs.~(\ref{eq:states})
and~(\ref{eq:measurements})    is    immaterial   for    the
formulation of  the inference protocol. The  only thing that
matters is that a ``special''  vector, with respect to which
the  arrow of  causality is  defined, is  fixed in  any real
space. Hence,  the problem of inference  considered here can
be formulated in a  completely basis-independent fashion. In
other  words, any  linear transformation  of the  underlying
linear spaces does not  affect the inference protocol (while
of course non-linear transformations  would not preserve the
structure of the underlying linear space).

\begin{figure}[t]
  \begin{overpic}[width=.8\columnwidth]{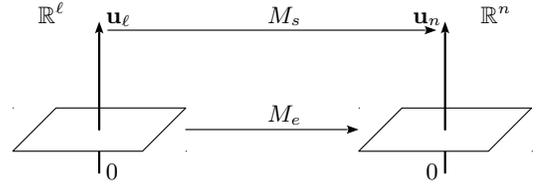}
    \put (18, 29) {$\mathbf{u}_\ell$}
    \put (77, 29) {$\mathbf{u}_n$}
    \put (18, -1) {$0$}
    \put (79.5, -1) {$0$}
    \put (49, 29) {$M_s$}
    \put (49, 10) {$M_e$}
    \put (5, 29) {$\R^\ell$}
    \put (90, 29) {$\R^n$}
  \end{overpic}
  \caption{The  linear space  $\R^\ell$ on  the left  is the
    state/effect   space,   the   vector   $\mathbf{u}_\ell$
    representing   the  unit   effect  and   the  hyperplane
    $\mbf{u}_\ell\cdot   \mbf{s}=1$  being   the  space   of
    states.  The  linear space  $\R^n$  on  the right  is  a
    probability space,  the vector $\mathbf{u}_n$  being the
    vector of  all ones  and the  hyperplane $\mbf{u}_n\cdot
    \mbf{p}=1$ defining probability  distributions. A family
    of states acts as a linear map $M_s \in \R^{\ell \to n}$
    mapping     $\mathbf{u}_\ell$    into     $\mathbf{u}_n$
    [Eq.~\eqref{eq:states}].  A family of  effects acts as a
    linear  map  $M_e  \in  \R^{\ell  \to  n}$  mapping  the
    hyperplane   of   states   into  that   of   probability
    distributions [Eq.~\eqref{eq:measurements}]. In fact, as
    noticed  in   the  main  text,  the   actual  choice  of
    coordinates    of    vectors    $\mathbf{u}_\ell$    and
    $\mathbf{u}_n$ is  immaterial for  the problem  at hand,
    which    can    be    formulated   in    a    completely
    basis-independent fashion.}
  \label{fig:spaces}
\end{figure}

\textit{Data-driven        inference}.        ---        Let
$M\in\mathbb{R}^{\ell\to n}$ be the linear map corresponding
to  a family  of states  (effects) of  a system  with effect
(state) space $\mbb{X}\subset\mathbb{R}^\ell$.  We denote by
$M\mbb{X}\subset\mathbb{R}^n$  the   image  of  $\mathbb{X}$
under $M$. Then, given  the data $\mcal{X} \subseteq\R^n$ as
a  set   of  probability  vectors,   we  say  that   $M$  is
\textit{consistent  with the  data} if  $\mcal{X}\subseteq M
\mbb{X}$. In words: there  exist elements of $\mbb{X}$ that,
acted  upon  by  transformation $M$,  give  the  probability
vectors  $\mcal{X}$. Among  all linear  maps $M$  consistent
with  the  data,  we  are interested  in  the  \textit{least
  committal}  ones.   Here,  we  quantify   the  ``committal
degree''  of  a  linear  map $M$  by  the  \textit{Euclidean
  volume}  of the  set  of probability  vectors  the map  is
consistent with.   This volume,  denoted by  $\vol(M \mbb{X}
)$, coincides with  the volume of the  \textit{range} of the
transformation  $M$~\cite{DBB17,  DBBV17, Dal17},  which  is
known  to  constitute  a  crucial  statistical  property  of
measurements~\cite{clean-POVMs}                          and
ensembles~\cite{quantum-blackwell}.  For example,  the range
of a  pair of  states coincides with  the Lorenz  region (or
testing region) of the  pair~\cite{renes, bus-gour}, and the
corresponding volume is  just the area of  region.  In order
to   avoid  comparing   volumes  of   sets  with   different
dimensionalities,  we   minimize  the  volume   over  linear
transformations  $M$ such  that  $M  \mbb{X} \subseteq  \spn
\mcal{X}$.

Presently we can define the main protocol:
\begin{dfn}[Data-driven inference]
  For  any  $\mcal{X}  \subseteq \R^{n}$  and  any  $\mbb{X}
  \subseteq \R^\ell$, we define
  \begin{align}
    \label{eq:ddi}
    \ddi_{s/e}   \left(   \mcal{X}   |\mbb{X}   \right)   :=
    \argmin_{M } \vol \left(M \mbb{X} \right),
  \end{align}
  where the  optimization is over  the linear maps  $M$ that
  satisfy
  \begin{align}
    \label{eq:consistency}
    \mcal{X} \subseteq M \mbb{X} \subseteq \spn \mcal{X}
  \end{align}
  and  either  Eq.~\eqref{eq:states}  for states  ($s$),  or
  Eq.~\eqref{eq:measurements} for effects ($e$). A pictorial
  sketch is given as Fig.~\ref{fig:inference}.
\end{dfn}

This  definition   should  clarify  that  our   approach  is
insensitive  to linear  transformations  of the  probability
space, as  any such transformation would  rescale the volume
of  any body  by a  constant  that uniquely  depends on  the
transformation, thus not affecting the output of data-driven
inference.

\begin{figure}[t]
  \begin{overpic}[width=.8\columnwidth]{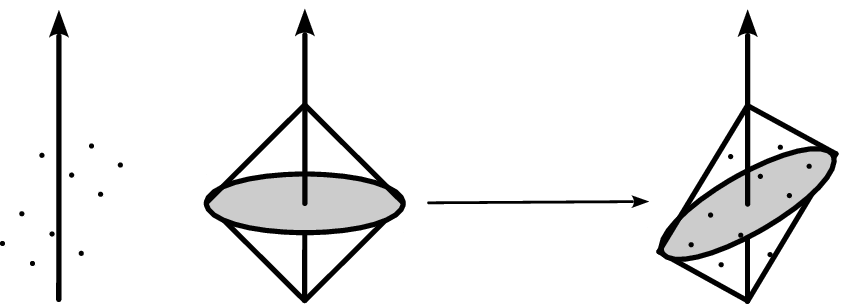}
    \put (8, 32) {$\mbf{u}_n$}
    \put (37, 32) {$\mbf{u}_\ell$}
    \put (89, 32) {$\mbf{u}_n$}
    \put (0, 32) {$\R^n$}
    \put (25, 32) {$\R^\ell$}
    \put (78, 32) {$\R^n$}
    \put (10, 8) {$\mcal{X}$}
    \put (40, 20) {$\mbb{X}_e^\ell$}
    \put (51, 14) {$\ddi_s (\mcal{X} | \mbb{X}_e^\ell)$}
  \end{overpic}
  \\\vspace{5mm}
  \begin{overpic}[width=.8\columnwidth]{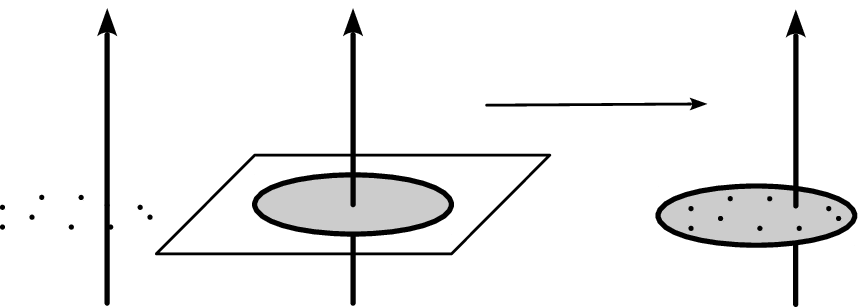}
    \put (14, 32) {$\mbf{u}_n$}
    \put (42, 32) {$\mbf{u}_\ell$}
    \put (94, 32) {$\mbf{u}_n$}
    \put (5, 32) {$\R^n$}
    \put (30, 32) {$\R^\ell$}
    \put (83, 32) {$\R^n$}
    \put (4, 6) {$\mcal{X}$}
    \put (35, 10) {$\mbb{X}_s^\ell$}
    \put (58, 25) {$\ddi_e (\mcal{X} | \mbb{X}_s^\ell)$}
  \end{overpic}
  \caption{\textbf{Top:} taking as input a set $\mcal{X}$ of
    probability vectors (represented as dots) and some prior
    information $\mbb{X}$  about the effect space  (the cone
    $\mbb{X}_e^\ell$    in    the     figure),    the    map
    $\ddi_s(\mcal{X} | \mbb{X} )$ returns the minimum volume
    linear   transformation  of   $\mbb{X}$  that   contains
    $\mcal{X}$, as per  Eq.~\eqref{eq:consistency}, and that
    satisfies    Eq.~\eqref{eq:states}.     \textbf{Bottom:}
    taking  as   input  a  set  $\mcal{X}$   of  probability
    distributions  (represented  as  dots)  and  some  prior
    information $\mbb{X}$ about the  state space (the sphere
    $\mbb{X}_s^\ell$    in    the     figure),    the    map
    $\ddi_e(\mcal{X} | \mbb{X} )$ returns the minimum volume
    linear   transformation  of   $\mbb{X}$  that   contains
    $\mcal{X}$, as per  Eq.~\eqref{eq:consistency}, and that
    satisfies Eq.~\eqref{eq:measurements}.}
  \label{fig:inference}
\end{figure}

\textit{Machine  learning of  states and  measurements}. ---
Given the convexity of  the merit function $\vol(M \mbb{X})$
and   of   the    constraints   in   Eqs.~\eqref{eq:states},
\eqref{eq:measurements},   and~\eqref{eq:consistency},   the
data-driven   inference   map   corresponds  to   a   convex
programming problem~\cite{BV04}.

Notice that,  in general,  the linear space  $\spn \mcal{X}$
can  be of  smaller dimension  than the  linear space  $\spn
\mbb{X}$.  In  this case, the optimization  over linear maps
$M$  that satisfy  Eq.~\eqref{eq:consistency}  can be  split
into:
\begin{itemize}
    \item[i)] the  optimization over a subspace  of the same
      dimension as $\spn \mcal{X}$, followed by
    \item[ii)]  an optimization  over linear  maps $M$  with
      such a subspace as its support.
\end{itemize}

In the case  when $M$ satisfies Eq.~\eqref{eq:measurements},
it  is  further clear  that  $\mbf{u}_\ell$  belongs to  the
support of  $M$.  However,  in the  case when  $M$ satisfies
Eq.~\eqref{eq:states},  $\mbf{u}_\ell$ does  not necessarily
belong to the support of $M$,  unless of course one has that
the  dimension of  $\spn \mcal{X}$  equals $\ell$,  in which
case  the  only possible  subspace  is  the space  $\R^\ell$
itself.      These     situations    are     depicted     in
Fig.~\ref{fig:supports}.

\begin{figure}[t]
  \begin{overpic}[width=.4\columnwidth]{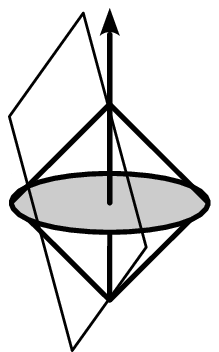}
    \put (52, 84) {$\mathbf{u}_\ell$}
    \put (0, 70) {$\supp M$}
    \put (70, 84) {$\R^\ell$}
  \end{overpic}
  \begin{overpic}[width=.4\columnwidth]{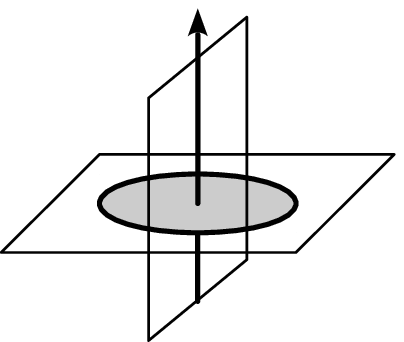}
    \put (52, 84) {$\mathbf{u}_\ell$}
    \put (14, 70) {$\supp M$}
    \put (70, 84) {$\R^\ell$}
  \end{overpic}
  \caption{\textbf{Left:}  conical effect  space around  the
    unit effect $\mbf{u}_\ell$. Any family of states acts as
    a linear  map $M$  whose support, solely  constrained by
    Eq.~\eqref{eq:states},  does   not  necessarily  contain
    $\mathbf{u}_\ell$.    \textbf{Right:}  spherical   state
    space on  the plane  orthogonal to  $\mbf{u}_\ell$.  Any
    measurement acts as a linear  map $M$ whose support, due
    to  Eq.~\eqref{eq:measurements},   necessarily  contains
    $\mathbf{u}_\ell$.}
  \label{fig:supports}
\end{figure}

Let  us consider  now the  case  when the  state and  effect
spaces, denoted  with $\mbb{X}_s^\ell\subset\mathbb{R}^\ell$
and       $\mbb{X}_e^\ell\subset\mathbb{R}^\ell$,       are,
respectively,  the (hyper)-sphere  in  the (hyper)-plane  of
states  orthogonal to  $\mbf{u}_\ell$, and  the (hyper)-cone
around $\mbf{u}_\ell$. This situation  occurs in the case of
classical  and quantum  bits,  with  $\ell=2$ and  $\ell=4$,
respectively.  Due  to the (hyper)-spherical  symmetry, when
inferring    a   measurement    $M_e$,    that   is,    when
Eq.~\eqref{eq:measurements} is satisfied,  the step i) above
corresponds     to    replacing     $\mbb{X}_s^\ell$    with
$\mbb{X}_s^m$, where  $m\le\ell$ is  the dimension  of $\spn
\mcal{X}$.  In  other words,  it  is  enough to  reduce  the
dimension   of   the   state   space,   while   keeping   it
(hyper)-spherical.  On the contrary, when inferring a set of
states  $M_s$,   that  is,  when   Eq.~\eqref{eq:states}  is
satisfied, an equivalent result does not hold: in this case,
the optimization  over the  support of  $M_s$ can  break the
(hyper)-conical symmetry of $\mbb{X}_e$.

For this reason, while  conceptually equivalent, the problem
of inferring  a measurement  is formally different  from the
problem of inferring a set  of states. As a consequence, the
machine learning algorithm that  we analytically develop and
discuss in Appendix~\ref{app:spherical},  while always valid
in  the case  of measurement  inference, can  be applied  to
states inference only when  the dimension of $\spn \mcal{X}$
equals $\ell$.

\textit{Observational completeness}.  --- Let  us now take a
step  backward  and consider  the  experiment  in which  the
dataset $\mcal{X}$ (we recall that $\mcal{X}$ is taken to be
a set of  probability vectors) is generated.  Upon the input
of  a  classical  variable  $i$, for  instance  through  the
pressure  of  a  button,   a  state  preparator  prepares  a
state. The  state is  then fed into  a measurement,  and the
outcome $j$  of the measurement,  which can be modeled  as a
light  bulb lighting  up,  is recorded.   The experiment  is
repeated ideally infinitely many  times, and the frequencies
are    estimated.     This     setup    is    depicted    in
Fig.~\ref{fig:setup}.

In  the protocol  of  \textit{data-driven reconstruction  of
  states}, a family of states $M_s$ acts on a set of effects
$X\subseteq\mathbb{X}_e$,   thus   producing   the   dataset
$\mathcal{X}=M_s X$. In this case, the experimentalist's aim
is to  choose the  ``probe'' measurement $X$  in such  a way
that the data-driven inference  applied to the corresponding
$\mcal{X}$  correctly outputs  the  range of  the family  of
states $M_s$ actually used in the experiment.

In complete analogy, in  the protocol of \textit{data-driven
  reconstruction of measurements}, the experimentalist's aim
is    to   choose    a    family    of   ``probe''    states
$X\subseteq\mathbb{X}_s$, such  that, once  measured through
$M_e$, a dataset $\mcal{X}=M_eX$  is produced, for which the
data-driven inference correctly outputs the range of $M_e$.

The  property  that such  probes  (states,  in the  case  of
measurement  inference;  effects,  in   the  case  of  state
inference) need  to satisfy  in order  that the  protocol of
data-driven  inference   always  succeeds,  is   defined  as
follows:

\begin{dfn}[Observational completeness]
  A set of effects $X \subseteq \mbb{X}_e \subseteq \R^\ell$
  is  observationally complete  for  a set  of states  $M_s$
  whenever
  \begin{align*}
    \ddi_{s} \left( M_sX |  \mbb{X}_e \right) = \left\{ M_s
    \mbb{X}_e \right\}.
  \end{align*}
  Analogoulsy,  a  set  of  states  $X  \subseteq  \mbb{X}_s
  \subseteq  \R^\ell$  is  observationally  complete  for  a
  measurement $M_e$ whenever
  \begin{align*}
    \ddi_{e} \left( M_eX |  \mbb{X}_s \right) = \left\{ M_e
    \mbb{X}_s \right\}.
  \end{align*}
\end{dfn}

In other words, an observationally complete set of states is
such that, when  fed through a measurement,  it provides the
same amount of statistical  information (for the protocol of
data-driven inference) as if the \textit{entire} state space
was measured. An  observationally complete measurement plays
the  same role  in the  inference of  states.  Observational
completeness hence  guarantees that the  maximum information
is  provided to  the inference  protocol. In  this case,  as
shown   in   Ref.~\cite{DBBT18},   the   reconstruction   of
$M\mbb{X}$ allows  for the identification of  the invertible
linear map $M$ up to gauge  symmetries (the case when $M$ is
not  invertible,   also  discussed   in  Ref.~\cite{DBBT18},
involves  more  technicalities),  that   is,  up  to  linear
transformations that  preserve $\mbb{X}$. This is  of course
the maximum level of accuracy that one should expect from an
inference protocol that only  relies on the bare coincidence
data.

\textit{Characterization of observational completeness}. ---
According to its  definition, the observational completeness
of  a set  $\mcal{X}$  depends  upon the  linear  map to  be
reconstructed.  However,  as it had already  been noticed in
Ref.~\cite{DBBT18},  such  a  dependency  turns  out  to  be
limited to  the support  of the linear  map, and  we discuss
here  a  few  important  consequences  of  this  fact.   Let
$\mcal{X}_0$ and  $\mcal{X}_1$ be  two subsets  of $\R^\ell$
related  by   an  invertible  transformation,  that   is  $M
\mcal{X}_0 =  \mcal{X}_1$.  The  following two  facts follow
immediately.  Whenever $M$ is a gauge symmetry, if either of
the two sets is observationally complete for $\R^\ell$, also
the other one  is.  If instead $M$ is not  a gauge symmetry,
then at  most one  between $\mcal{X}_0$ and  $\mcal{X}_1$ is
observationally complete  for $\R^\ell$, but not  both. This
situation is depicted in Fig.~\ref{fig:ocness}.

\begin{figure}[t]
  \begin{overpic}[width=.8\columnwidth]{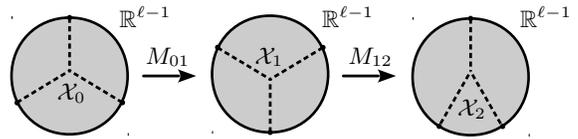}
    \put (21, 21) {$\R^{\ell-1}$}
    \put (60, 21) {$\R^{\ell-1}$}
    \put (98, 21) {$\R^{\ell-1}$}
    \put (26, 14) {$M_{01}$}
    \put (65, 14) {$M_{12}$}
    \put (9, 7) {$\mcal{X}_0$}
    \put (47, 14) {$\mcal{X}_1$}
    \put (86, 4) {$\mcal{X}_2$}
  \end{overpic}
  \caption{The set  $\mbb{X}$ of states is  represented by a
    grey  circle.  Sets  $\mcal{X}_0$  and $\mcal{X}_1$  are
    related  by a  gauge symmetry  (a $\pi$-rotation)  hence
    either both of  them or none of  them is observationally
    complete (in this case, the  former is the case as shown
    in the main text, since  regular simplices are spherical
    $2$  designs).  Sets  $\mcal{X}_1$ and  $\mcal{X}_2$ are
    related by a  linear map which is not  a gauge symmetry,
    hence at most one among them is observationally complete
    (in this case, $\mcal{X}_1$).}
  \label{fig:ocness}
\end{figure}

A closed-form characterization of observational completeness
can be derived for the cases of (hyper)-conical effect space
and  (hyper)-spherical  state  space.    In  this  case,  by
extending  John's theory~\cite{Joh48}  on extremum  problems
with  inequalities  as  subsidiary conditions,  we  show  in
Appendix~\ref{app:spherical}     a      relation     between
observational completeness and spherical designs.

Operationally (for a formal  definition of spherical design,
see Appendix~\ref{app:design}), a spherical $t$-design is an
ensemble  $\{ p_k,  \mbf{v}_k  \}$ (that  is, a  probability
distribution  $p_x$   over  states  $\mbf{v}_k$)   which  is
indistinguishable from  the uniform ensemble over  states on
the  boundary of  the  (hyper)-sphere, when  $t$ copies  are
given.   We  say  that  a set  $\{  \mbf{v}_k  \}  \subseteq
\R^\ell$  supports  a  $t$-design whenever  there  exists  a
probability  distribution $\{  p_k  \}$ such  that $\{  p_k,
\tilde{\mbf{v}}_k    \}$    is     a    $t$-design,    where
$\tilde{\mbf{v}}_k  :=  (\mbf{u}_\ell \cdot  \mbf{v}_k)^{-1}
\mbf{v}_k$ lie on the (hyper)-plane of states.

We have  then the following closed-form  characterization of
observational completeness for  systems with (hyper)-conical
effect   space  or   (hyper)-spherical  state   space.   Let
$\mcal{X}$ be a set of  states or effects, that is $\mcal{X}
\subseteq    \mbb{X}_s^\ell$    or    $\mcal{X}    \subseteq
\mbb{X}_e^\ell$,   respectively.   If   set  $\mcal{X}$   is
observationally complete  for an invertible linear  map $M$,
then  $\mcal{X}$  supports  a  spherical  $2$  design.   The
generalization   of   this   statement  to   the   case   of
non-invertible linear map  $M$ involves some technicalities,
and is therefore deferred to Appendix~\ref{app:spherical}.

If  $\mcal{X}$  is  a  set  of  states,  that  is  $\mcal{X}
\subseteq \mbb{X}_s^\ell$, also the vice-versa is true. That
is,  if $\mcal{X}$  supports  a  spherical $2$-design,  then
$\mcal{X}$  is observationally  complete for  any invertible
linear map  $M$. Again,  the generalization  to the  case of
non-invertible    linear   map    $M$    is   deferred    to
Appendix~\ref{app:spherical}.    We  conjecture   a  similar
result to  hold if $\mcal{X}$ is  a set of effects,  that is
$\mcal{X} \subseteq \mbb{X}_e^\ell$.

The following two facts follow as immediate corollaries. The
minimum cardinality observationally complete set for a qubit
is the symmetric, informationally complete set. As a further
corollary, the minimum  cardinality observationally complete
set  of basis  for a  qubit  system are  the three  mutually
unbiased  bases.  These  result  provide  a new  operational
interpretation to these sets, based on data-driven inference
rather than on their purely mathematical definition in terms
of equiangular vectors.

\textit{Conclusion}.    ---  Data-driven   inference  is   a
protocol  that,  upon the  input  of  a set  of  probability
vectors, outputs the mathematical description for a physical
device. Such  a description is self-consistent,  that is, it
can generate the given probability vectors.  Moreover, it is
minimally  committal, that  is,  it is  consistent with  the
minimal set of probability vectors.

In  this  work, we  provided  a  unified formalism  for  the
data-driven inference  in the  cases where  the mathematical
description  is in  terms  of states  and measurements.  For
systems    with    (hyper)-conical     effect    space    or
(hyper)-spherical   state  space,   we  provided   a  convex
programming algorithm for the machine learning of states and
measurements based on data-driven inference.

Observational completeness is the  property of any apparatus
that, when applied to a target device, generates probability
vectors  for  which  the  output  of  data-driven  inference
coincides  with  the range  of  the  device itself.   Hence,
observational completeness  plays for  data-driven inference
the  same  role  played by  informational  completeness  for
conventional tomography.

In  this  work,  we  provided  a  full  characterization  of
observational  completeness.   Our  characterization  is  in
closed-form for systems with (hyper)-conical effect space or
(hyper)-spherical state space,  in which cases observational
completeness for  a set implies  that such a set  supports a
spherical $2$-design. We showed  that the vice-versa is true
for sets  of states, and  we conjectured  it to be  the case
also   for   sets   of  effects.    Accordingly,   symmetric
informationally  complete sets  and mutually  unbiased bases
are  minimal cardinality  observationally  complete sets  of
vectors   and   bases,   respectively.    We   conclude   by
conjecturing  that   for  arbitrarily   dimensional  quantum
systems, quantum  $2$-designs coincide  with observationally
complete sets.

\textit{Acknowledgement}.  --- This work is supported by the
National  Research  Fund  and  the  Ministry  of  Education,
Singapore,   under  the   Research  Centres   of  Excellence
programme;  and   partly  supported   by  the   program  for
FRIAS-Nagoya  IAR Joint  Project Group.   F.~B. acknowledges
partial support from the Japan  Society for the Promotion of
Science (JSPS) KAKENHI, Grant No. 19H04066.

\appendix

\setcounter{dfn}{0}
\setcounter{thm}{0}

\section{An Example}\label{app:example}

As an  example, we  consider a source  that can  produce the
three pure states of a real qubit
\begin{align*}
  \rho_1   &  =   \frac12  \left(   \openone  +   \sigma_{z}
  \right),\\   \rho_2   &   =   \frac12   \left(\openone   +
  \frac{\sqrt{3}}2                \sigma_{x}               -
  \frac12\sigma_{z}\right),\\  \rho_3  &  =  \frac12  \left(
  \openone   -   \frac{\sqrt{3}}2   \sigma_{x}   -   \frac12
  \sigma_{z}\right),
\end{align*}
and a measurement described by the effects
\begin{align*}
  E_1 & = \frac14  \left(\openone + \sigma_{z}\right),\\ E_2
  & = \frac14 \left(\openone  - \sigma_{z}\right),\\ E_3 & =
  \frac14  \left(\openone  +  \sigma_{x}\right),\\ E_4  &  =
  \frac14 \left(\openone - \sigma_{x}\right).
\end{align*}
It is easy to check that this example gives rise to the data
given in Eq.~\eqref{dataexample} of the main text.

First  let  us consider  the  case  of \textit{inference  of
  measurements}. Each state $\rho_i$  has associated with it
the vector $\mbf{x}_i$ where $\left(\mbf{x}_{i}\right)_{j} =
p_{ij}=  P\left(E_{j}|\rho_{i}\right)$.  Therefore,  we will
have 3 points in $\mbb{R}^{4}$:
\begin{align*}
  \mbf{x}_1   &  =   \left[\frac12,   0,  \frac14,   \frac14
    \right]^T,\\  \mbf{x}_2  &   =  \left[\frac18,  \frac38,
    \frac{2+\sqrt{3}}8,                   \frac{2-\sqrt{3}}8
    \right]^T,\\  \mbf{x}_3  &   =  \left[\frac18,  \frac38,
    \frac{2-\sqrt{3}}8, \frac{2+\sqrt{3}}8\right]^T.
\end{align*}
These   points    are   in   a   2-dimensional    plane   in
$\mbb{R}^{4}$.  In this  plane, any  measurement defines  an
ellipsoid as the set of all  the vectors it can produce. The
measurement being  used must  of course define  an ellipsoid
that contains  the three  observed points, and  $\ddi$ finds
the  consistent ellipsoid  with  the  smallest volume.   The
inferred range is  then inverted to give the  effects, up to
symmetries.

Then we  consider the case of  \textit{inference of states}.
This  time,  to  each   effect  one  associates  the  vector
$\mbf{x}_i$ where $(\mbf{x}_{j})_i = p_{ij}$.  Thus, we will
now have 4 points in $\mbb{R}^{3}$:
\begin{align*}
  \mbf{x}_1                        &                       =
  \left[\frac12,\frac18,\frac18\right]^T,\\  \mbf{x}_2  &  =
  \left[0,\frac38,\frac38\right]^T,\\    \mbf{x}_3    &    =
  \left[\frac14,\frac{2+\sqrt{3}}8,\frac{2-\sqrt{3}}8\right]^T,\\ \mbf{x}_4
  &                                                        =
  \left[\frac14,\frac{2-\sqrt{3}}8,\frac{2+\sqrt{3}}8\right]^T.
\end{align*}
The next  step now is  to find the linear  transformation of
the space of effects -  that preserves the null and identity
effects -  that contains all  four points and is  of minimal
volume. This volume is then  inverted to find the states (up
to symmetries) that induce this linear transformation of the
space of effects.

\section{Formalization}\label{app:formalization}

In these appendices, for compactness the subscripts $s$ and $e$ adopted in the main text are replaced by $+$ and $-$, respectively.

For  any  $\mcal{X}  \subseteq  \R^n$ let  us  define  $\vol
(\mcal{X})$ as  the Euclidean volume of  $\mcal{X}$ on $\aff
\mcal{X}$.  One immediately has
\begin{align}
  \label{eq:volume}
  \vol    \left(M   \mbb{X}\right)    =    \left|   M^T    M
  \right|_+^{\frac12} \vol \left( M^+ M \mbb{X} \right).
\end{align}

Let  us introduce  a family  $\{ \mbf{u}_n  \in \R^n  \}$ of
vectors  and two  families  $\mcal{M}_\pm^{\ell  \to n}$  of
linear transformations
\begin{align*}
  \mcal{M}_+^{\ell \to n} & := \left\{ M \in \R^{\ell \to n}
  \;    \Big|     \;    M    \mbf{u}_\ell     =    \mbf{u}_n
  \right\},\\  \mcal{M}_-^{\ell \to  n} &  := \left\{  M \in
  \R^{\ell \to n}  \; \Big| \; M^T  \mbf{u}_n = \mbf{u}_\ell
  \right\}.
\end{align*}
Notice that if $M \in  \mcal{M}_+^{\ell \to n}$ one has $M^+
M \mbf{u}_\ell  \neq 0$ and  if $M \in  \mcal{M}_-^{\ell \to
  n}$ one  has $M^+ M \mbf{u}_\ell  = \mbf{u}_\ell$.  Notice
also that  if $M_0 \in  \mcal{M}_\pm^{\ell \to n}$  and $M_1
\in   \mcal{M}_\pm^{n  \to   m}$  one   has  $M_1   M_0  \in
\mcal{M}_\pm^{\ell \to  m}$. Notice  finally that if  $M \in
\mcal{M}_\pm^{\ell  \to n}$  and $M$  is invertible  one has
$M^{-1} \in \mcal{M}_\pm^{n \to \ell}$.

For  any  $\mbb{X}  \subseteq  \R^\ell$  and  any  $\mcal{X}
\subseteq \R^n$ let us define
\begin{align*}
  \mcal{L} \left(  \mcal{X} |  \mbb{X} \right) :=  \left\{ M
  \in  \R^{\ell   \to  n}  |  \mcal{X}   \subseteq  M\mbb{X}
  \subseteq \spn\mcal{X} \right\},
\end{align*}
and let $\mcal{L}_\pm  ( \mcal{X} | \mbb{X} )  := \mcal{L} (
\mcal{X} | \mbb{X} ) \cap \mcal{M}_\pm^{\ell \to n}$.

\begin{dfn}[Data-driven inference]
  \label{def:inference}
  For  any $\mbb{X}  \subseteq  \R^\ell$  and any  $\mcal{X}
  \subseteq \R^n$, let us define
  \begin{align*}
    \ddi_\pm \left( \mcal{X}  |\mbb{X} \right) := \argmin_{M
      \in  \mcal{L}_\pm \left(  \mcal{X} |  \mbb{X} \right)}
    \vol \left(M \mbb{X} \right).
  \end{align*}
\end{dfn}

\begin{dfn}[Observational completeness]
  \label{def:ocness}
  Any given  $\mcal{X} \subseteq \mbb{X}  \subseteq \R^\ell$
  is  OC with  respect to  $\mbb{X}$  for any  given $L  \in
  \mcal{M}_\pm^{\ell \to n}$ if and only if
  \begin{align*}
    \ddi_\pm \left( L\mcal{X} |  \mbb{X} \right) = \left\{ L
    \mbb{X} \right\}.
  \end{align*}
\end{dfn}

\section{General results}\label{app:general}

For  any  $\mbb{X}  \subseteq  \R^\ell$  and  any  $\mcal{X}
\subseteq \R^n$, let us define
\begin{align*}
  \Pi_\pm   \left(    \mcal{X}   |   \mbb{X}    \right)   :=
  \argmin_{\substack{\Pi = \Pi^2\\\rank \Pi = m}} \vol\left(
  \ddi_\pm \left( \mcal{X} | \Pi \mbb{X} \right) \right),
\end{align*}
where $m := \dim \spn \mcal{X}$. One immediately has
\begin{align*}
  \ddi_\pm   \left(   \mcal{X}   |   \mbb{X}   \right)   =
  \bigcup\limits_{\Pi  \in \Pi_\pm(  \mcal{X} |  \mbb{X})}
  \ddi_\pm \left( \mcal{X} | \Pi \mbb{X} \right).
\end{align*}

By   explicit  computation,   for  any   $\mbb{X}  \subseteq
\R^\ell$,  any $\mcal{X}  \subseteq  \R^n$, and  any $L  \in
\mcal{M}_\pm$ such that $L^+ L \mbb{X} = \mbb{X}$ one has
\begin{align*}
  \ddi_\pm  \left( \mcal{X}  |  \mbb{X}  \right) =  \ddi_\pm
  \left( \mcal{X} | L \mbb{X} \right).
\end{align*}

\begin{lmm}[Commutativity]
  \label{lmm:commutativity}  
  For  any   $\mbb{X}  \subseteq  \R^\ell$,   any  $\mcal{X}
  \subseteq \R^n$,  and any  $L \in \mcal{M}_\pm$  such that
  $L^+ L \mcal{X} = \mcal{X}$ one has
  \begin{align}
    \label{eq:commutativity0}
    \ddi_\pm  \left(  \mcal{X} |  \mbb{X}  \right)  & =  L^+
    \ddi_\pm      \left(       L      \mcal{X}|      \mbb{X}
    \right),\\  \label{eq:commutativity1} L  \ddi_\pm \left(
    \mcal{X}  |  \mbb{X}  \right)  &  =  \ddi_\pm  \left(  L
    \mcal{X}| \mbb{X} \right).
  \end{align}
\end{lmm}

\begin{proof}
  By      direct      computation     $L^+      \mcal{L}_\pm
  (L\mcal{X}|\mbb{X})   \subseteq  \mcal{L}_\pm   (\mcal{X}|
  \mbb{X})$  and   $L  \mcal{L}_\pm  (\mcal{X}   |  \mbb{X})
  \subseteq \mcal{L}_\pm  (L \mcal{X}|\mbb{X})$.  Hence $L^+
  \mcal{L}_\pm  (L   \mcal{X}  |  \mbb{X})   =  \mcal{L}_\pm
  (\mcal{X}  | \mbb{X})$  and  $L  \mcal{L}_\pm (\mcal{X}  |
  \mbb{X}) = \mcal{L}_\pm (L\mcal{X} | \mbb{X})$. Hence
  \begin{align*}
    \ddi_\pm  \left( \mcal{X}|\mbb{X}  \right) =  \argmin_{M
      \in \mcal{L}_\pm \left( L  \mcal{X} | \mbb{X} \right)}
    f\left(L^+M \mbb{X} \right).
  \end{align*}
  
  Since     $\dim    \spn     \mbb{X}     =    \ell$,     by
  Definition~\ref{def:inference} for any $M \in \mcal{L}_\pm
  (L \mcal{X}|\mbb{X})$ one has that  $MM^+ \le LL^+$ is the
  projector on  $\spn L \mcal{X}$.   Hence $| (L^+  M)^T L^+
  M|_+ = | M |_+^2 | (L^+ M M^+)^T L^+ M M^+ |_+$.  Hence by
  Definition~\ref{def:inference} one has
  \begin{align*}
    \argmin_{M  \in  \mcal{L}_\pm  \left( L  \mcal{X}  |  \mbb{X}
      \right)}  f\left(L^+M,  \mbb{X}\right) =  L^+  \ddi_\pm
    \left( L \mcal{X}|\mbb{X} \right).
  \end{align*}

  Thus Eq.~\eqref{eq:commutativity0} follows.   Since $L^+ L
  \ddi_\pm  (\mcal{X}  |  \mbb{X}) =  \ddi_\pm  (\mcal{X}  |
  \mbb{X})$, Eq.~\eqref{eq:commutativity1} follows.
\end{proof}

\begin{thm}[Data-driven inference]
  \label{thm:inference}
  Let  $\mbb{X} \subseteq  \R^\ell$ and  $\mcal{X} \subseteq
  \R^n$  and  $m  :=  \dim \supp  \mcal{X}$.   For  any  any
  $\mcal{M} \subseteq  \mcal{M}_\pm^{\ell \to m}$  such that
  $\supp \mcal{L} =  \Pi_\pm (\mcal{X} | \mbb{X}  )$ and any
  $L \in \mcal{M}_\pm^{n  \to m}$ such that $\supp  L = \spn
  \mcal{X}$ , one has
  \begin{align*}
    \ddi_\pm  \left(  \mcal{X}  |   \mbb{X}  \right)  =  L^+
    \bigcup\limits_{M  \in   \mcal{M}}  \ddi_\pm   \left(  L
    \mcal{X} | M \mbb{X} \right).
  \end{align*}
\end{thm}

\begin{proof}
  The  statement directly  follows from  the application  of
  Lemma~\ref{lmm:commutativity}.
\end{proof}

\begin{thm}[Observational completeness]
  \label{thm:ocness}
  Let $\mcal{X} \subseteq \mbb{X}  \subseteq \R^\ell$ and $m
  := \dim  \supp \mcal{X}$.  One  has that $\mcal{X}$  is OC
  for  $N \in  \mcal{M}_\pm^{\ell  \to n}$  with respect  to
  $\mbb{X}$   if   and  only   if   there   exists  $L   \in
  \mcal{M}_\pm^{\ell  \to  m}$ with  $L^+  L  = N^+  N$  and
  $\mcal{M} \subseteq \mcal{M}_\pm^{\ell \to m}$ with $\supp
  \mcal{M} = \Pi_\pm (\mcal{X} | \mbb{X} )$ such that
  \begin{align*}
    \bigcup\limits_{M \in \mcal{M}}  \ddi_\pm ( L \mcal{X}_1
    | M \mcal{X}_0 ) = \left\{ L \mcal{X}_0 \right\}.
  \end{align*}
\end{thm}

\begin{proof}
  The  statement directly  follows from  the application  of
  Lemma~\ref{lmm:commutativity}.
\end{proof}

\section{(Hyper)-spherical case}\label{app:spherical}

For any $\mbf{v} \in \R^\ell$, upon defining
\begin{align*}
  g\left( \mbf{v} \right) = \left| \mbf{v} \right|_2 -
  \sqrt{2} \hat{\mbf{u}}_\ell \cdot \mbf{v},
\end{align*}
one   has    that   the   (hyper)-spherical    state   space
$\mbb{X}_-^\ell$  and   the  (hyper)-conical   effect  space
$\mbb{X}_+^\ell$ are given by
\begin{align}
  \label{eq:sphere}
  \mbb{X}_-^\ell  &  := \left\{  \mbf{v}  \;  \Big| \;  g\left(
  \mbf{v}  \right)   \le  0,  \;   \mbf{u}_\ell  \cdot
  \mbf{v} = 1 \right\},\\
  \label{eq:cone}
  \mbb{X}_+^\ell  &   :=  \left\{  \mbf{v}  \;   \Big|  \;
  g\left(\mbf{v}\right) \le 0, \; g\left( \mbf{u}_\ell
  - \mbf{v}\right) \le 0 \right\}.
\end{align}

\begin{cor}[Data-driven inference]
  \label{cor:inference}
  For any $L \in \mcal{M}_\pm^{\ell \to m}$ such that $\supp
  L = \spn \mcal{X}$, one has
  \begin{align*}
    \ddi_\pm \left( \mcal{X} | \mbb{X}^\ell_\pm \right) & = L^+
    \ddi_\pm \left( L \mcal{X}| \mbb{X}^m_\pm \right),
  \end{align*}
  where $m := \dim \supp L$, for any $m$ in the $-$ case and
  for $m = \ell$ in the $+$ case.
\end{cor}

\begin{proof}
  The       statement       directly      follows       from
  Theorem~\ref{thm:inference}   and   Eqs.~\eqref{eq:sphere}
  and~\eqref{eq:cone}.
\end{proof}

\begin{cor}[Observational completeness]
  \label{cor:ocness}
  Any $\mcal{X} \subseteq \mbb{X}_\pm^\ell$ is OC for $N \in
  \mcal{M}_\pm^{\ell \to n}$ if and  only if there exists $L
  \in \mcal{M}_\pm^{\ell \to  m}$ with $L^+ L =  N^+ N$ such
  that $L \mbb{X}_\pm^\ell = \mbb{X}_\pm^m$ such that
  \begin{align*}
    \ddi_\pm   (  L   \mcal{X}|  \mbb{X}^m_\pm   )  =   \left\{
    \mbb{X}^m_\pm \right\},
  \end{align*}
  where $m := \dim \supp L$.
\end{cor}

\begin{proof}
  The       statement       directly      follows       from
  Theorem~\ref{thm:ocness}     and    Eqs.~\eqref{eq:sphere}
  and~\eqref{eq:cone}.
\end{proof}

A set  $\mcal{X}$ is $\mbf{u}_\ell/2$-symmetric if  and only
if for  any $\mbf{v} \in  \mcal{X}$ one has  $\mbf{u}_\ell -
\mbf{v} \in \mcal{X}$.  Clearly  the set $\mbb{X}_+^\ell$ is
$\mbf{u}_\ell/2$-symmetric.

\begin{lmm}
  For any  invertible $M  \in \mcal{M}_\pm^{\ell  \to \ell}$
  and    any    $\mcal{X}     \in    \mbb{X}_-$    or    any
  $\mbf{u}_\ell/2$-symmetric $\mcal{X}  \in \mbb{X}_+^\ell$,
  the following are equivalent conditions:
  \begin{enumerate}
    \item\label{item:consistency0}  $\mcal{X}   \subseteq  M
      \mbb{X}_\pm^\ell$,
    \item\label{item:consistency1}  $g(M^{-1}
      \mbf{v}) \le 0$, for any $\mbf{v} \in \mcal{X}$.
  \end{enumerate}
\end{lmm}

\begin{proof}
  Due   to   the  invertibility   of   $M$   one  has   that
  condition~\eqref{item:consistency0}   is   equivalent   to
  $M^{-1}    \mcal{X}    \in   \mbb{X}_\pm^\ell$.     Hence,
  implication      $\ref{item:consistency0}      \Rightarrow
  \ref{item:consistency1}$    follows    immediately    from
  Eqs.~\eqref{eq:sphere}   and~\eqref{eq:cone}.   To   prove
  implication      $\ref{item:consistency1}      \Rightarrow
  \ref{item:consistency0}$,  we  need   to  distinguish  two
  cases.

  Let us first consider the case $\mbb{X}_-^\ell$.  Since by
  hypothesis   $M   \in    \mcal{M}_-^\ell$,   by   explicit
  computation  one has  $M^{-1} \in  \mcal{M}_-^\ell$. Hence
  for    any    $\mbf{v}    \in    \mcal{X}_-$    one    has
  $\hat{\mbf{u}}_\ell \cdot M^{-1} \mbf{v}  = 1$.  Hence the
  implication remains proved.

  Let us  now consider  the case $\mbb{X}_+^\ell$.   For any
  $\mbf{v}  \in \mcal{X}$  by hypothesis  one has  $g(M^{-1}
  \mbf{v}) \le 0$.  Due  to the $\mbf{u}_\ell/2$-symmetry of
  $\mcal{X}$,  also $\mbf{u}_\ell  - \mbf{v}  \in \mcal{X}$,
  from  which   by  hypothesis  $g(M^{-1}   (\mbf{u}_\ell  -
  \mbf{v}))  \le 0$.   Since $g(M^{-1}  \mbf{v}) \le  0$ and
  $g(M^{-1}   (\mbf{u}_\ell   -   \mbf{v}))   \le   0$,   by
  Eq.~\eqref{eq:cone}   one   has    $M^{-1}   \mbf{v}   \in
  \mbb{X}_+^\ell$.  Hence, the implication remains proved.
\end{proof}

From    Eq.~\eqref{eq:volume}     one    has     $\vol    (M
\mbb{X}_\pm^\ell) \propto \sqrt{f(M)}$, where
\begin{align*}
  f \left( M \right) := \log \left| M^T M \right|.
\end{align*}

The  constraint  $M  \in \mcal{M}_\pm^{\ell  \to  \ell}$  in
$\ddi(L  \mcal{X} |  \mbb{X}_\pm^\ell)$ can  be implemented  by
introducing the auxiliary functions:
\begin{align*}
  h^\pm   \left(   N  \right)   =   \Pi^\mp   N  \Pi^\pm   +
  \hat{\mbf{u}}_\ell^{\otimes 2}, \qquad \Pi^\pm := \openone
  - \frac{1\pm1}2 \hat{\mbf{u}}_\ell^{\otimes 2}.
\end{align*}

By direct inspection $f(h^\pm(N))$ and $g(h^\pm(N) \mbf{v})$
are convex functions of $N \in \R^{n \to n}$.  Hence for any
$\mcal{X} \subseteq \R^\ell$ one has that $\ddi_\pm(\mcal{X}
| \mbb{X}_\pm^\ell)$ is  a convex programming  problem, that
can be  efficiently solved in  $N$.  To this aim,  one needs
the  Jacobian   and  Hessian   matrices  (with   respect  to
$\operatorname{vec}(N)$) of $f$ and $g$. From the chain rule
it immediately  follows that  for any function  $g: \R^{\ell
  \to \ell} \to \R$ one has
\begin{align}
  \label{eq:chain0} \jac   g   \circ   h^\pm   (N)  &   =   \left.    \jac   g
  \right|_{h^\pm\left(N\right)}        \Pi^\pm       \otimes
  \Pi^\mp,\\ \label{eq:chain1}  \hes g \circ h^\pm  \left( N
  \right)  &  =  \Pi^\pm  \otimes  \Pi^\mp  \left.   \hes  g
  \right|_{h^\pm\left(N\right)} \Pi^\pm \otimes \Pi^\mp.
\end{align}
By explicit computation one has
\begin{align*}
    \jac   f\left(M^{-1}\right)   &   =   -2   M^T,\\   \hes
    f\left(M^{-1}\right) & = 2 M^T \otimes M S_{\ell^2},
\end{align*}
where  $S_{\ell^2}$  denotes the  $\ell^2$-dimensional  swap
operator, and
\begin{align*}
  \jac   g  \left(M^{-1}   \mbf{v}   \right)   &  =   \left(
  \frac{M^{-1} \mbf{v}}{ \left|  M^{-1} \mbf{v} \right|_2} -
  \sqrt{2}      \hat{\mbf{u}}_\ell      \right)      \otimes
  \mbf{v},\\ \hes g\left( M^{-1}  \mbf{v} \right) & = \left|
  M^{-1} \mbf{v}\right|_2^{-1} \left(  \openone_{\ell^2} - 2
  \frac{\left(M^{-1}    \mbf{v}\right)^{\otimes   2}}{\left|
    M^{-1}     \mbf{v}     \right|_2}    \right)     \otimes
  \mbf{v}^{\otimes 2}.
\end{align*}

\begin{thm}
  If  a   set  $\mcal{X}  \subseteq  \mbb{X}_-^\ell$   or  a
  $\mbf{u}_\ell/2$-symmetric    set   $\mcal{X}    \subseteq
  \mbb{X}_+^\ell$ is OC for a given $M \in \R^{\ell \to n}$,
  then $L \mcal{X}$ supports a spherical $2$-design, for any
  $L \in \mcal{L}_\pm^{\ell \to m}$ with $m := \dim \supp M$
  such  that $M^+  M  =  L^+ L$  and  $L \mbb{X}_\pm^\ell  =
  \mbb{X}_\pm^m$.
\end{thm}

\begin{proof}
  Due to Corollary~\ref{cor:ocness}  one has that $\mcal{X}$
  is  OC for  $M$  if  and only  if  $\ddi_\pm(L \mcal{X}  |
  \mbb{X}_\pm^m)   =   \{    \mbb{X}_\pm^m   \}$.   Due   to
  Lemma~\ref{lmm:john}  there exists  $\{  \lambda_k \ge  0,
  \mbf{v}_k    \in    L    \mathcal{X}   \}$    such    that
  Eq.~\ref{eq:john} holds when computed in $\openone_m$.  By
  explicit computation one has
  \begin{align*}
    \left.   \jac  f  \left(  M^{-1}  \right)  \right|_{M  =
      \openone_m} = -2 \openone_m,
  \end{align*}
  and
  \begin{align*}
    \left|  \mbf{v}  \right|_2^{-1}  \left.  \jac  g  \left(
    M^{-1}  \mbf{v}  \right)   \right|_{M  =  \openone_m}  =
    \frac12  \tilde{\mbf{v}}^{\otimes  2} -  \hat{\mbf{u}}_m
    \otimes \tilde{\mbf{v}},
  \end{align*}
  where     $\tilde{\mbf{v}}_k     =    (\mbf{u}_m     \cdot
  \mbf{v}_k)^{-1} \mbf{v}_k$.  By defining $p_k := (4 \ell -
  4)^{-1} |\mbf{v}_k|_2 \lambda_k$ by Eqs.~\eqref{eq:chain0}
  and~\eqref{eq:chain1}    one    has    that    $\{    p_k,
  \tilde{\mbf{v}}_k \}$ is a spherical $2$-design, hence the
  implication follows.
\end{proof}

\begin{thm}
  If a set $\mcal{X}  \subseteq \mbb{X}_-^\ell$ is such that
  $L \mcal{X}$ supports a  spherical $2$-design, for some $L
  \in    \mcal{L}_\pm^{\ell   \to    m}$   such    that   $L
  \mbb{X}_\pm^\ell  =  \mbb{X}_\pm^m$  for  some  $m$,  then
  $\mcal{X}$ is OC for any $M \in \R^{\ell \to n}$ such that
  $m := \dim \supp M$ such that $M^+ M = L^+ L$.
\end{thm}

\begin{proof}
  Due to Corollary~\ref{cor:ocness}  one has that $\mcal{X}$
  is  OC for  $M$  if  and only  if  $\ddi_\pm(L \mcal{X}  |
  \mbb{X}_-^m) =  \{ \mbb{X}_-^m  \}$. By  hypothesis, there
  exixsts a  probability distribution $\{ p_k  \}$ such that
  $\{ p_k,  \mbf{v}_k \in  \mcal{X} \}$  is a  spherical $2$
  design. Hence,  for any linear  map $M$ such  that $M^{-1}
  \mcal{X} \subseteq \mbb{X}_-$ one has
  \begin{align*}
    0 \ge \sum_k p_k g \left( M^{-1} v_k \right).
  \end{align*}
  By using Eq.~\eqref{eq:sph2des} one immediately has
  \begin{align*}
    & \sum_k p_k  g \left( M^{-1} v_k \right) \\  = & \frac{
      \Tr\left[ M^{-1} {M^{-1}}^T \right] + \left( \ell -2 \right) \left( \left| M^{-1} \hat{u}
      \right|_2^2 -2 \right)     -2    \left|     {M^{-1}}^T    \hat{u}
      \right|_2^2}{|\mbf{u}_\ell|_2^2      \left(\ell     -1
      \right)}.
  \end{align*}
  By using  Eq.\eqref{eq:measurements} and the fact  that $|
  M^{-1} \hat{u} |_2^2 \ge 1$ one has
  \begin{align*}
    \sum_k   p_k   g   \left(   M^{-1}   v_k   \right)   \ge
    \frac{\Tr\left[    M^{-1}     {M^{-1}}^T    \right]    -
      \ell}{|\mbf{u}_\ell|_2^2 \left(\ell -1 \right)}.
  \end{align*}
  Since for  any $X \ge  0$ one  has $\Tr[X -  \openone] \ge
  \log | X  |$ with equality if and only  if $X = \openone$,
  one has $\log |M^{-1} {M^{-1}}^T| \le 0$, with equality if
  and  only if  $M$  is an  orthogonal  matrix.  Hence,  the
  statement remains proved.
\end{proof}

\section{Spherical $t$-designs}\label{app:design}

\begin{dfn}[Spherical $t$-design]
  A  probability distribution  $\{ p_k  \}$ over  states $\{
  \mbf{v}_k \in \R^\ell \}$, that  is $\{ p_k, \mbf{v}_k \}$
  such that $\mbf{v}_k \cdot \mbf{u}_\ell  = 1$ for any $k$,
  is a spherical $t$-design if and only if
  \begin{align*}
    \sum_k p_k \mbf{v}_k^{\otimes t} = \int \dif O \left(
    O \mbf{v} \right)^{\otimes t},
  \end{align*}
  where $\dif O$ denotes the  Haar measure of the orthogonal
  representation  of   the  symmetries  of   $\mbb{X}_-$  and
  $\mbf{v}$ is any vector on the boundary of $\mbb{X}_-$.
\end{dfn}

A set  $\{ \mbf{v}_k  \in \R^\ell  \}$ supports  a spherical
$t$-design if  there exists  a probability  distribution $\{
p_k  \}$  such that  $\{  p_k,  \tilde{\mbf{v}}_k \}$  is  a
spherical    $t$-design,    where   $\tilde{\mbf{v}}_k    :=
(\mbf{u}_\ell \cdot \mbf{v}_k)^{-1} \mbf{v}_k$ for any $k$.

Here we  consider spherical  $2$-designs. When  working with
spherical $2$-design,  for any  $\mbf{v} \in \R^\ell$  it is
convenient to  adopt the convention $\mbf{v}^{\otimes  2} :=
\mbf{v} \mbf{v}^T$. Then, by explicit computation one has
\begin{align*}
  \int  \dif O  \; \left(  O \mbf{v}  \right)^{\otimes 2}  =
  \frac1{\left|\mbf{u}_\ell\right|^2_2}\left( \frac1{\ell-1}
  \openone_\ell            +           \frac{\ell-2}{\ell-1}
  \hat{\mbf{u}}_\ell^{\otimes 2} \right).
\end{align*}
By multiplying both sides by $\mbf{u}_\ell$ on the right one
has
\begin{align*}
  \int      \dif       O      \;      O       \mbf{v}      =
  \frac{\mbf{u}_\ell}{\left|\mbf{u}_\ell\right|_2^2}.
\end{align*}

Hence,  any  $\{  p_k,  \mbf{v}_k \in  \mbb{S}  \}$  is  a
spherical $2$-design if and only if it satisfies
\begin{align}
  \label{eq:sph2des}
  \sum_k       p_k       \mbf{v}_k^{\otimes       2}       =
  \frac1{\left|\mbf{u}_\ell\right|_2^2}\left(\frac1{\ell-1}
  \openone_\ell            +           \frac{\ell-2}{\ell-1}
  \hat{\mbf{u}}_\ell^{\otimes 2} \right),
\end{align}
in which case it is also a spherical $1$-design, that is, it
satisfies
\begin{align*}
  \sum_k   p_k    \mbf{v}_k   =   \frac{\mbf{u}_\ell}{\left|
    \mbf{u}_\ell \right|_2^2}.
\end{align*}

\section{John's extremality conditions}\label{app:john}

Let $f(L)$ be some differentiable function, let
\begin{align*}
  \mbb{X}  :=  \left\{ \mbf{v}  \in  \R^\ell  \; \Big|  \;
  g(\mbf{v}) \ge 0 \right\},
\end{align*}
for some differentiable $g : \R^\ell \to \R$ and let $\mcal{X}
\subseteq \R^\ell$.  For any $L  \in \R^{\ell \to \ell}$ and
any $\{ \lambda_k \ge 0, \mbf{v}_k \in \mcal{X} \}$ let
\begin{align*}
  h  \left(  L,  \left\{  \lambda_k,  \mbf{v}_k  \right\}
  \right) := f(L) + \sum_k \lambda_k g(L \mbf{v}_k).
\end{align*}

\begin{lmm}[John's necessary condition]
  \label{lmm:john}
  For some  $L^* \in \mcal{M}_\pm^{\ell \to  \ell}$, one has
  that  $L^*  \mbb{X}  \in  \ddi_\pm  (  \mcal{X}  \;  |  \;
  \mbb{X})$ implies that there exists $\{ \lambda_k^* \ge 0,
  \mbf{v}_k^* \in \mcal{X} \}$ such that
  \begin{align}
    \label{eq:john}
    \left.     \frac{\partial     h(L,    \{    \lambda_k^*,
      \mbf{v}_k^* \})}{\partial L} \right|_{L = L^*} = 0.
  \end{align}
\end{lmm}

\begin{proof}
  Theorem I of Ref.~\cite{Joh48}.
\end{proof}

\begin{lmm}[John's sufficient condition]
  If there exists $\{  \lambda_k^* \ge 0, \mbf{v}_k^* \in
  \mathcal{X}  \}$ such  that Eq.~\eqref{eq:john}  holds and
  the set
  \begin{align*}
    \dim \spn \left\{  \left.  \frac{\partial f(L)}{\partial
      L}  \right|_{L =  L^*},  \left.   \frac{\partial g(  L
      \mbf{v}_k   )}{\partial   L}  \right|_{L   =   L^*}
    \right\}_k = \ell \left( \ell -1 \right),
  \end{align*}
  for some $L^* \in \mcal{M}_\pm^{\ell \to \ell}$, then $L^*
  \mbb{X} \in \ddi_\pm ( \mcal{X} |\mbb{X} )$.
\end{lmm}

\begin{proof}
  Theorem II of Ref.~\cite{Joh48}.
\end{proof}

\end{document}